\documentclass[12pt,a4paper]{article}
\usepackage[margin=1in]{geometry}
\usepackage[T1]{fontenc}
\usepackage[utf8]{inputenc}
\usepackage{amsmath,amssymb,amsthm}
\usepackage{booktabs}
\usepackage[hidelinks]{hyperref}

\newtheorem{theorem}{Theorem}[section]
\newtheorem{proposition}[theorem]{Proposition}
\newtheorem{lemma}[theorem]{Lemma}

\theoremstyle{remark}
\newtheorem{remark}[theorem]{Remark}

\newcommand{\Def}{\mathrm{Def}}
\newcommand{\Ric}{\mathrm{Ric}}
\newcommand{\R}{\mathbb{R}}
\newcommand{\C}{\mathbb{C}}
\newcommand{\HH}{\mathbb{H}}
\newcommand{\PP}{\mathbb{P}}
\newcommand{\LD}{\Delta_{\Def}}
\newcommand{\dv}{\operatorname{div}}
\newcommand{\re}{\operatorname{Re}}
\newcommand{\im}{\operatorname{Im}}
\newcommand{\Lp}{L^p_\sigma}
\newcommand{\ip}[2]{\langle #1,#2\rangle}

\title{The $L^p$ spectrum of the Stokes operator\\
on hyperbolic 3-space and exponential decay\\
for the Navier-Stokes equations}
\author{Zhi-Wei Wang${}^{1,2}$ and
Samuel L.\ Braunstein${}^{2,*}$\\[6pt]
\small ${}^{1}$College of Physics, Jilin University,
Changchun 130012, China\\
\small ${}^{2}$Computer Science, University of York,
York YO10 5GH, UK\\[6pt]
\small E-mail: zhiweiwang.phy@gmail.com (Z.-W.W.),
sam.braunstein@york.ac.uk (S.L.B.)\\
\small ${}^{*}$Corresponding author}
\date{\today}

\begin{document}
\maketitle

\begin{abstract}
We determine the $L^p$ spectrum, $1<p<\infty$, of the Stokes operator
$A=-\PP\LD$ on divergence-free vector fields of hyperbolic 3-space of
curvature $-a^2$, where $\LD=\Delta_B+\Ric$ is the Ebin-Marsden deformation
Laplacian and $\PP$ is the Leray projector. The spectrum is the parabolic
region $a^2\{4+z^2:\ z\in\C,\ |\im z|\le|1-2/p|\}$. For $p>2$ every interior
point is an eigenvalue, realized by an explicit family of divergence-free
swirl fields which leave $L^p$ exactly at the boundary of the region. It
follows that the Stokes semigroup on $\Lp$ has exponential type exactly
$-\mu_p$ with $\mu_p=a^2\bigl(4-(1-2/p)^2\bigr)$. At $p=3$ this is
$35a^2/9$, which lies strictly above the rates $26a^2/9$ and $10a^2/3$ that
follow from pointwise domination by the scalar heat semigroup and from
interpolation, and strictly below the $L^2$ gap $4a^2$. For the
Navier-Stokes equations with small data in $L^3$, whose global solvability
on $\HH^3$ is due to Pierfelice, the solutions decay in $L^3$ and in $L^6$
at every rate below $\mu\cdot35a^2/9$, where $\mu$ is the viscosity, and no
faster rate holds for the linearized flow. The exponent $3$ is the usual
scaling threshold. Curvature fixes the time scale of the decay and
introduces no new threshold.
\end{abstract}

\section{Introduction}
\label{sec:intro}

On a Riemannian manifold $(M,g)$ the incompressible Navier-Stokes equations
require a choice of viscous operator acting on vector fields. Ebin and
Marsden~\cite{EM70} identified the deformation tensor of a viscous fluid
with the Lie derivative $\mathcal L_ug$, which leads to the operator
\begin{equation}\label{eq:def}
\LD u=\Delta_Bu+\Ric\,u,
\end{equation}
where $\Delta_B=\operatorname{tr}\nabla^2$ is the rough Laplacian. The
choice is debated~\cite{Czu24,CCD17}. In~\cite{WB-sel} we argued for
\eqref{eq:def} on kinematic grounds, and in~\cite{WB-shell} we showed that
it is the operator produced by stress-free wall conditions in the
thin-shell limit of a viscous layer around a hypersurface, whereas
vorticity-free conditions produce the Hodge Laplacian; the two differ by
$2\Ric$, which on $\HH^3$ is the constant $-4a^2$. With this operator and viscosity
$\mu>0$ the equations read
\begin{equation}\label{eq:NS}
\partial_tu+\nabla_uu-\mu\LD u+\nabla p=0,\qquad\dv u=0,
\end{equation}
and their mild formulation is
\begin{equation}\label{eq:mild}
u(t)=e^{-t\mu A}u_0-\int_0^te^{-(t-s)\mu A}\,\PP\nabla\!\cdot(u\otimes u)(s)\,ds,
\qquad A:=-\PP\LD,
\end{equation}
where $\PP$ is the Leray projector onto divergence-free fields and $A$ is
the Stokes operator.

Pierfelice~\cite{Pie17} developed the Fujita-Kato theory~\cite{FK64,Kat84}
for~\eqref{eq:NS} on complete, simply connected, non-compact manifolds with
bounded geometry, two-sided negative bounds on the Ricci curvature and
non-positive sectional curvature. She dominated the semigroup $e^{t\LD}$
pointwise by a scalar heat semigroup and obtained $L^p$-$L^q$ estimates with
exponential decay \cite[Lemma~4.5, Theorem~4.1]{Pie17}. When the metric is
Einstein the Leray projector commutes with $\LD$, and she proved global
existence for small data in $L^n$ \cite[Theorem~5.2]{Pie17}, in a solution
space that carries the exponential decay of these estimates;
Balentine~\cite{Bal20} extended Kato's $L^p$
theory to $\HH^n$ by the same estimates, and periodic and almost periodic
solutions on Einstein manifolds were studied in~\cite{HXHM21,HH24arch}. In
dimension three the Weyl tensor vanishes, so an Einstein metric has constant
sectional curvature, and these results concern hyperbolic space $\HH^3$
alone among complete simply connected 3-manifolds.

The decay rate that these estimates give on $\HH^3$ of curvature $-a^2$,
for data in $L^3$, is $26a^2/9$. It comes from the domination
\begin{equation}\label{eq:dom}
|e^{t\LD}\omega|\le e^{-2a^2t}\,e^{t\Delta_0}|\omega|
\end{equation}
by the scalar heat semigroup $e^{t\Delta_0}$, whose $L^3$ operator norm
decays at the rate $8a^2/9$. The present paper answers the question of what
the exact rate is. The answer is $35a^2/9$, and more generally the Stokes
semigroup on divergence-free $L^p$ fields has exponential type exactly
\begin{equation}\label{eq:mup}
-\mu_p,\qquad \mu_p:=a^2\Bigl(4-\bigl(1-\tfrac2p\bigr)^2\Bigr)
=a^2\,\frac{(p+2)(3p-2)}{p^2},\qquad 1<p<\infty .
\end{equation}
Thus the domination~\eqref{eq:dom} loses exactly $a^2$ for every $p$.

The input is the $L^p$ spectrum of the Hodge Laplacian on differential
forms over hyperbolic space, computed by Charalambous and Lu~\cite{CL25}.
On $\HH^3$ the Stokes operator is $\Delta_H+4a^2$ on coclosed 1-forms
(Section~\ref{sec:prelim}), so what is needed is the $L^p$ spectrum of
$\Delta_H$ on coclosed forms. The result of~\cite{CL25} concerns all
1-forms and gives one inclusion. It also shows that for $p>2$ every interior
point of the spectral region is an eigenvalue of $\Delta_H$ on $L^p$
1-forms, but the Stokes operator acts only on coclosed forms. For the other
inclusion we exhibit, for each
spherical harmonic of degree one, a one-parameter family of explicit
coclosed eigenforms of $\Delta_H$ on $\HH^3$, with eigenvalue $z^2$,
$z\in\C$, which belong to $L^p$ precisely when $p>2$ and
$|\im z|<1-2/p$. As vector fields they are rotationally symmetric swirls
about an axis. Their eigenvalues fill the interior of the spectral region,
and they leave $L^p$ exactly where the exponential type is attained.

We now state the results. Throughout, $\HH^3$ has sectional curvature
$-a^2<0$, $\Lp$ is the space of divergence-free fields in $L^p$, and
$A_p$ denotes the Stokes operator on $\Lp$ (Section~\ref{sec:prelim}).
For $1\le p\le\infty$ put
\begin{equation}\label{eq:Q}
Q_p:=\bigl\{z^2:\ z\in\C,\ |\im z|\le|1-\tfrac2p|\bigr\}.
\end{equation}
$Q_p$ is the closed region to the right of the parabola
$\{x^2-(1-2/p)^2+2ix(1-2/p):x\in\R\}$. Its leftmost point is
$-(1-2/p)^2$, its interior is $\{z^2:|\im z|<|1-2/p|\}$, and $Q_2=[0,\infty)$.

\begin{theorem}[the $L^p$ spectrum of the Stokes operator]\label{thm:spec}
Let $1<p<\infty$. Then
\begin{equation}\label{eq:spec}
\sigma(A_p)=a^2\,(4+Q_p).
\end{equation}
For $p>2$ every interior point $a^2(4+z^2)$, $|\im z|<1-2/p$, is an
eigenvalue of $A_p$, with the eigenfield $u_z$ of
Lemma~\ref{lem:swirl} below. For $1<p\le2$ the operator $A_p$ has no
eigenvalues.
\end{theorem}

\begin{theorem}[exact $L^p$ decay rates]\label{thm:rates}
Let $1<p<\infty$ and let $S_p(t)=e^{-tA_p}$ be the Stokes semigroup on
$\Lp$. Then for all $t\ge0$ and every $\epsilon>0$,
\begin{equation}\label{eq:rates}
e^{-\mu_pt}\le\|S_p(t)\|_{\Lp\to\Lp}\le C_\epsilon\,e^{-(\mu_p-\epsilon)t},
\end{equation}
with $\mu_p$ as in~\eqref{eq:mup}. In particular
$\lim_{t\to\infty}t^{-1}\log\|S_p(t)\|_{\Lp\to\Lp}=-\mu_p$, and
$\mu_3=35a^2/9$. For $p=2$ one may take $\epsilon=0$ and $C_0=1$.
\end{theorem}

The value $\mu_p$ is symmetric under $p\leftrightarrow p'$, equals $4a^2$
at $p=2$, and decreases to $3a^2$ as $p\to1$ or $p\to\infty$. Table~\ref{tab:rates}
compares it at $p=3$ with the lower bounds available by soft arguments.

\begin{table}[ht]
\centering\small
\begin{tabular}{lll}
\toprule
Bound on the $L^3$ decay rate of $S_3(t)$ & Value & Source\\
\midrule
domination by the scalar semigroup & $26a^2/9\approx2.89\,a^2$ &
\cite[Thm.~4.1]{Pie17}; Prop.~\ref{prop:lower}(i)\\
interpolation between $L^2$ and $L^\infty$ & $10a^2/3\approx3.33\,a^2$ &
Prop.~\ref{prop:lower}(ii)\\
exact exponential type & $35a^2/9\approx3.89\,a^2$ &
Theorem~\ref{thm:rates}\\
$L^2$ spectral gap, for comparison & $4a^2$ & \cite{Don81};
Prop.~\ref{prop:lower}(iii)\\
\bottomrule
\end{tabular}
\caption{Decay rates of the Stokes semigroup on $L^3_\sigma(\HH^3)$,
curvature $-a^2$.}
\label{tab:rates}
\end{table}

For the Navier-Stokes equations the consequence is the following. We state
it for $a=1$; Section~\ref{sec:scaling} gives the general case by scaling.

\begin{theorem}[decay of small solutions]\label{thm:NS}
Let $\HH^3$ have curvature $-1$, let $\mu>0$, and let $0<\lambda<35/9$. Put
\begin{equation}\label{eq:X}
X_\lambda:=\Bigl\{u\in C((0,\infty);L^6_\sigma):\
\|u\|_{X_\lambda}:=\sup_{t>0}\min(\mu t,1)^{1/4}e^{\mu\lambda t}\|u(t)\|_{L^6}
<\infty\Bigr\}.
\end{equation}
There are constants $c_\lambda,C_\lambda,C'_\lambda>0$, depending only on
$\lambda$, such that for every $u_0\in L^3_\sigma$ with
$\|u_0\|_{L^3}<\epsilon_0:=c_\lambda\mu$ the equation~\eqref{eq:mild} has
exactly one solution in the ball
$\{u\in X_\lambda:\|u\|_{X_\lambda}\le C_\lambda\|u_0\|_{L^3}\}$, and this
solution satisfies, for all $t>0$,
\begin{equation}\label{eq:NSdecay}
\|u(t)\|_{L^6}\le C_\lambda\|u_0\|_{L^3}\min(\mu t,1)^{-1/4}e^{-\mu\lambda t},
\qquad
\|u(t)\|_{L^3}\le C'_\lambda\|u_0\|_{L^3}\,e^{-\mu\lambda t}.
\end{equation}
The threshold $35/9$ is sharp for the linear part. If $\lambda>35/9$ there
is no constant $C$ with $\|e^{-t\mu A}u_0\|_{L^3}\le Ce^{-\mu\lambda t}
\|u_0\|_{L^3}$ for all $u_0\in L^3_\sigma$ and $t\ge0$.
\end{theorem}

Global existence and uniqueness for small $L^3$ data on $\HH^3$ are
contained in \cite[Theorem~5.2]{Pie17} and in~\cite{Bal20}, with the decay
rate $26/9$ of \cite[Theorem~4.1]{Pie17}. The content of
Theorem~\ref{thm:NS} is the rate and its sharpness.

\paragraph{Relation to earlier versions of this preprint.}
Earlier versions called $26/9$ the effective $L^3$ spectral gap of the
deformation Laplacian on $\HH^3$ and used it as the decay rate, and the third
version described the resulting rate $26\mu/9$ as optimal. Only the lower
bound $26/9$, which follows from the domination~\eqref{eq:dom}, was proved
there, and the rate is not optimal. The exact $L^3$ rate is $35/9$.
Donnelly's theorem~\cite{Don81}, which those versions cited for the $L^2$
gap $4a^2$, is an $L^2$ statement and gives no information on $L^3$. The
present version is organized around the exact rates,
and the global existence statement of the earlier versions is credited
to~\cite{Pie17}. The companion paper~\cite{WB-pinched} treats complete
3-manifolds with pinched negative curvature, where the Leray projector no
longer commutes with $\LD$ and only lower bounds for the rate are
available; the value $35a^2/9$ is the benchmark against which those bounds
are to be read.

\paragraph{Organization.}
Section~\ref{sec:prelim} collects the facts about $\HH^3$, the Leray
projector and the semigroup that we use. Section~\ref{sec:spec} constructs
the swirl eigenfields and proves Theorem~\ref{thm:spec}.
Section~\ref{sec:rates} proves Theorem~\ref{thm:rates} and compares the
exact rate with the lower bounds. Section~\ref{sec:NS} proves
Theorem~\ref{thm:NS}. Section~\ref{sec:scaling} discusses scaling, the
critical exponent, and open questions.

\section{Preliminaries}
\label{sec:prelim}

\subsection{Hyperbolic space and normalization}

Multiplying the metric by a constant $\kappa^2>0$ divides the sectional
curvature by $\kappa^2$ and divides every operator of second order in
$\nabla$ by $\kappa^2$. Hence all spectral statements below are proved for
curvature $-1$, and the general case follows by multiplying eigenvalues and
rates by $a^2$. We use two models of $\HH^3$ with curvature $-1$. The upper
half-space model is $\{(x_1,x_2,y):y>0\}$ with
$g=y^{-2}(dx_1^2+dx_2^2+dy^2)$. Geodesic polar coordinates $(r,\theta,\phi)$
about a point give $g=dr^2+\sinh^2r\,(d\theta^2+\sin^2\theta\,d\phi^2)$.

On functions $\Delta_0=\dv\operatorname{grad}\le0$. The bottom of the $L^2$
spectrum of $-\Delta_0$ on $\HH^3$ is $\lambda_0=1$, and the $L^p$ spectrum
of $-\Delta_0$ is the region $1+Q_p$ \cite{DST88}, so that
\begin{equation}\label{eq:scalar}
\|e^{t\Delta_0}\|_{L^p\to L^p}=e^{-\lambda_pt},\qquad
\lambda_p:=\frac{4(p-1)}{p^2}=1-\bigl(1-\tfrac2p\bigr)^2,\qquad1\le p\le\infty .
\end{equation}
On every complete manifold the inequality $\le$ in~\eqref{eq:scalar} holds
with $\lambda_p$ replaced by $\lambda_0\lambda_p$, where $\lambda_0$ is the
bottom of the $L^2$ spectrum (see \cite[Prop.~4.6]{Pie17}); on $\HH^3$ it is
an equality because $-\lambda_p$ belongs to the $L^p$ spectrum of
$\Delta_0$, so that $e^{-\lambda_pt}\in\sigma(e^{t\Delta_0})$ by the
spectral inclusion theorem \cite[Thm.~IV.3.6]{EN00}.

\subsection{Operators on 1-forms}

We identify vector fields with 1-forms by the metric. On 1-forms,
$\Delta_H=d\delta+\delta d\ge0$ is the Hodge Laplacian, with
$\delta=-\dv$ on 1-forms, and the Weitzenb\"ock identity reads
$\Delta_H=-\Delta_B+\Ric$. On $\HH^3$ with curvature $-1$ we have
$\Ric=-2g$, hence on all 1-forms
\begin{equation}\label{eq:LD}
\LD=\Delta_B+\Ric=-\Delta_H-4 .
\end{equation}
The semigroup $e^{-t\Delta_H}$ on 1-forms is given by a smooth kernel with
Gaussian bounds, which is symmetric in the sense that
$\ip{e^{-t\Delta_H}\omega}{\eta}=\ip{\omega}{e^{-t\Delta_H}\eta}$ for
$\omega\in L^p$, $\eta\in L^{p'}$, where
$\ip{\omega}{\eta}=\int\langle\omega,\eta\rangle_g\,dV$ is the bilinear
pairing, taken without complex conjugation since the eigenforms below are
complex. The semigroup defines consistent strongly
continuous semigroups on $L^p(\Lambda^1)$ for $1\le p<\infty$, self-adjoint
on $L^2$ \cite[\S4]{Pie17}, \cite[Lemma~2.1 and \S3]{CL25}. We write
$\Delta_H^{(p)}$ for
minus its generator on $L^p$. Since the semigroups on $L^p$ and $L^{p'}$ are
adjoint to one another, so are their generators:
\begin{equation}\label{eq:green}
\ip{\Delta_H^{(p)}\omega}{\eta}=\ip{\omega}{\Delta_H^{(p')}\eta},\qquad
\omega\in D(\Delta_H^{(p)}),\ \eta\in D(\Delta_H^{(p')}).
\end{equation}
Compactly supported smooth forms belong to $D(\Delta_H^{(p)})$, on which
$\Delta_H^{(p)}$ acts classically. The forms we use below are smooth but not
compactly supported, and Lemma~\ref{lem:green} admits them.

\begin{lemma}[admissible forms]\label{lem:green}
Let $1<p<\infty$ and let $\omega$ be a smooth $1$-form on $\HH^3$ with
$\omega$, $\nabla\omega$ and $\Delta_H\omega$ in $L^p(\Lambda^1)$. Then
$\omega\in D(\Delta_H^{(p)})$ and $\Delta_H^{(p)}\omega=\Delta_H\omega$, so
that~\eqref{eq:green} applies to $\omega$.
\end{lemma}

\begin{proof}
Fix $o\in\HH^3$, write $r$ for the distance to $o$, take
$\chi\in C^\infty([0,\infty))$ with $\chi=1$ on $[0,1]$, $\chi=0$ on
$[2,\infty)$, and put $\chi_R:=\chi(r/R)$ for $R\ge1$, so that
$\chi_R\omega\in C_c^\infty(\Lambda^1)\subset D(\Delta_H^{(p)})$. The
Leibniz rule for the second-order operator $\Delta_H$ gives
$$\bigl|\Delta_H(\chi_R\omega)-\chi_R\Delta_H\omega\bigr|\le
C\bigl(|\mathrm{grad}\,\chi_R|\,|\nabla\omega|
+|\Delta_0\chi_R|\,|\omega|\bigr),$$
every term on the right carrying a derivative of $\chi_R$. Now
$|\mathrm{grad}\,r|=1$ and $\Delta_0r=2\coth r$, which is bounded on
$\{r\ge1\}$, so $|\mathrm{grad}\,\chi_R|\le CR^{-1}$ and
$|\Delta_0\chi_R|\le CR^{-1}$, both supported in $\{R\le r\le2R\}$. Since
$\omega,\nabla\omega\in L^p$, the right-hand side tends to $0$ in $L^p$,
while $\chi_R\omega\to\omega$ and $\chi_R\Delta_H\omega\to\Delta_H\omega$ in
$L^p$ by dominated convergence. The operator $\Delta_H^{(p)}$ is closed, so
$\omega\in D(\Delta_H^{(p)})$ with $\Delta_H^{(p)}\omega=\Delta_H\omega$.
\end{proof}

The domination~\eqref{eq:dom}, which for curvature $-1$ reads
\begin{equation}\label{eq:dom1}
|e^{t\LD}\omega|\le e^{-2t}\,e^{t\Delta_0}|\omega|\qquad\text{almost everywhere,}
\end{equation}
is \cite[Lemma~4.5]{Pie17}; its proof uses Kato's inequality, the maximum
principle and the upper Ricci bound $\Ric\le-2g$.\footnote{Her hypothesis
(H2) writes the Ricci bounds as $-c_0^{-1}g\le\Ric\le-c_0g$, which on
$\HH^3$ of curvature $-a^2$ requires $c_0\le\min\{2a^2,1/(2a^2)\}$. Since
the problem is invariant under the rescaling of
Section~\ref{sec:scaling}, one may assume $a^2\le1/2$, and then
$c_0=2a^2$ is admissible. In any case the proof of her Lemma~4.5 uses only
the upper bound.} Since $e^{t\Delta_0}$ is contractive on $L^1$ and on $L^\infty$, the semigroup
$e^{t\LD}$ is contractive on $L^1$ and on $L^\infty$, hence on every $L^p$ by
interpolation, and it is self-adjoint on $L^2$. These three properties give
analyticity.

\begin{lemma}[analyticity on $L^p$]\label{lem:analytic}
For $1<p<\infty$ the semigroup $e^{t\LD}$ extends to a contraction semigroup
on $L^p(\Lambda^1)$ on a sector about $(0,\infty)$ of half-angle
$\pi(1-\frac1{\min(p,p')})>0$. In particular it is a bounded analytic
semigroup on $L^p(\Lambda^1)$, and so is its part in any closed invariant
subspace.
\end{lemma}

\begin{proof}
On $L^2$ the operator $\LD$ is self-adjoint with $\LD\le-4$, so $e^{w\LD}$ is
defined by the spectral theorem and $\|e^{w\LD}\|_{2\to2}\le1$ whenever
$\re w\ge0$. Let $1<p<2$, fix $t>0$ and $0<\psi<\pi/2$, and for $\zeta$ in the
strip $0\le\re\zeta\le1$ put $F(\zeta):=e^{w(\zeta)\LD}$ with
$w(\zeta):=te^{i\psi\zeta}$. Then $|\arg w(\zeta)|\le\psi<\pi/2$ throughout the
strip, so for $f,h\in L^1\cap L^\infty$ the function
$\zeta\mapsto\ip{F(\zeta)f}{h}$ is analytic in the interior and bounded by
$\|f\|_2\|h\|_2$ on the closure. On $\re\zeta=0$ the time $w(\zeta)$ is
positive real, so $\|F(\zeta)\|_{1\to1}\le1$; on $\re\zeta=1$ we have
$\re w(\zeta)>0$, so $\|F(\zeta)\|_{2\to2}\le1$. Stein interpolation
\cite[Lemma~A.1]{CL25} at $\zeta=\vartheta$, where
$\frac1p=(1-\vartheta)+\frac\vartheta2$ and hence
$\vartheta=2(1-\frac1p)\in(0,1)$, gives
$\|e^{te^{i\psi\vartheta}\LD}\|_{p\to p}\le1$. Letting $t$ run over
$(0,\infty)$ and $\psi$ over $(0,\pi/2)$ gives the stated sector. The case
$p>2$ follows by duality, since $e^{t\LD}$ is self-adjoint, and the case of a
closed invariant subspace is immediate.
\end{proof}

\begin{remark}
The conclusion is the standard one for symmetric contraction
semigroups~\cite{Cow83}. We give the proof because Stein's
theorem~\cite[Ch.~III]{Ste70} is stated for symmetric diffusion semigroups,
which are in addition positivity preserving and conservative, and neither
property is available on $1$-forms.
\end{remark}

\subsection{The Leray projector and the Stokes operator}

Let $\PP:=I-d(-\Delta_0)^{-1}\delta$. On $L^2$ it is the orthogonal
projection onto the closed subspace of coclosed 1-forms.

\begin{lemma}\label{lem:P}
For $1<p<\infty$, $\PP$ extends to a bounded projection on $L^p(\Lambda^1)$,
symmetric with respect to the pairing of $L^p$ with $L^{p'}$. Its range
$\Lp:=\PP L^p$ is the space of $\omega\in L^p$ with $\delta\omega=0$ in the
sense of distributions, and $\PP$ commutes with $e^{-t\Delta_H}$ on $L^p$.
\end{lemma}

\begin{proof}
The Riesz transform $R:=d(-\Delta_0)^{-1/2}$ is bounded on $L^p(\HH^3)$ for
$1<p<\infty$ \cite{Str83,Loh85}, and $(-\Delta_0)^{-1/2}=
\pi^{-1/2}\int_0^\infty t^{-1/2}e^{t\Delta_0}dt$ is bounded on $L^p$
by~\eqref{eq:scalar}. On $L^2\cap L^p$ we have $I-\PP=RR^*$ with
$R^*=(-\Delta_0)^{-1/2}\delta$ bounded on $L^p$ as the adjoint of $R$ on
$L^{p'}$, which gives the boundedness and the symmetry. For $\omega\in L^p$
coclosed and $f\in C_c^\infty$,
$\ip{R^*\omega}{f}=\ip{\omega}{d(-\Delta_0)^{-1/2}f}=0$, because
$C_c^\infty(\HH^3)$ is dense in $\{h\in L^{p'}:dh\in L^{p'}\}$, by the cutoff
argument of Lemma~\ref{lem:green} applied to functions (see
also~\cite[Ch.~3]{Heb99}); so $\PP\omega=\omega$;
conversely $\delta\PP\omega=\delta\omega-\delta d(-\Delta_0)^{-1}\delta\omega
=0$ in the sense of distributions, which identifies the range.

For the commutation let $\omega\in C_c^\infty(\Lambda^1)$ and
$f:=(-\Delta_0)^{-1}\delta\omega$. Since $d\Delta_H=\Delta_Hd$,
$\delta\Delta_H=\Delta_H\delta$ and $\Delta_H=-\Delta_0$ on functions,
$\Delta_H(I-\PP)\omega=\Delta_Hdf=d(-\Delta_0f)=d\delta\omega$ and
$(I-\PP)\Delta_H\omega=d(-\Delta_0)^{-1}\delta\Delta_H\omega=d\delta\omega$.
So $\PP$ commutes with $\Delta_H$ on the core $C_c^\infty$, hence with
$\Delta_H^{(2)}$, hence with $e^{-t\Delta_H}$ on $L^2$, and by density and
consistency on $L^p$.
\end{proof}

By~\eqref{eq:LD} and Lemma~\ref{lem:P}, the Stokes operator
$-\PP\LD=\PP(\Delta_H+4)$ acts on $\Lp$ as $\Delta_H+4$. We therefore
define the Stokes operator $A_p$ on $\Lp$ as the part of
$\Delta_H^{(p)}+4$ in the invariant subspace $\Lp$, that is
$D(A_p)=D(\Delta_H^{(p)})\cap\Lp$ and $A_p\omega=\Delta_H^{(p)}\omega+4\omega$.
Its semigroup is
\begin{equation}\label{eq:S}
S_p(t):=e^{-tA_p}=e^{-4t}e^{-t\Delta_H}\big|_{\Lp}=e^{t\LD}\big|_{\Lp},
\end{equation}
a strongly continuous analytic semigroup on $\Lp$ by
Lemma~\ref{lem:analytic}. For curvature
$-a^2$ the same holds with $4a^2$ in place of $4$. With viscosity $\mu$ the
semigroup in~\eqref{eq:mild} is $S_p(\mu t)$.

\section{The \texorpdfstring{$L^p$}{Lp} spectrum of the Stokes operator}
\label{sec:spec}

\subsection{Swirl eigenfields}

In polar coordinates let $\eta:=\sin^2\theta\,d\phi$. It is the pull-back
of a coclosed 1-form on $S^2$, namely $\pm\star_{S^2}d\cos\theta$, and it
satisfies $\Delta_{S^2}\eta=2\eta$. For $z\in\C$ define
\begin{equation}\label{eq:phiz}
\varphi_z(r):=z\cos(zr)-\coth(r)\sin(zr)\quad(z\ne0),\qquad
\varphi_0(r):=1-r\coth r .
\end{equation}

\begin{lemma}[swirl eigenfields]\label{lem:swirl}
Let $z\in\C$. The function $\varphi_z$ is smooth and even on $\R$, with
$\varphi_z(r)=-\tfrac13z(z^2+1)\,r^2+O(r^4)$ for $z\ne0$ and
$\varphi_0(r)=-\tfrac13r^2+O(r^4)$, and it solves
\begin{equation}\label{eq:ode}
-\varphi''+\frac{2}{\sinh^2r}\,\varphi=z^2\varphi .
\end{equation}
It vanishes identically only for $z=\pm i$. The 1-form
$\omega_z:=\varphi_z(r)\,\eta$ is smooth on $\HH^3$, coclosed, and
\begin{equation}\label{eq:eig}
\Delta_H\omega_z=z^2\omega_z,\qquad\text{equivalently}\qquad
-\LD u_z=(4+z^2)\,u_z,
\end{equation}
where $u_z:=\varphi_z(r)\sinh^{-2}r\,\partial_\phi$ is the vector field
dual to $\omega_z$. Moreover
\begin{equation}\label{eq:normz}
|\omega_z|_g=\frac{|\varphi_z(r)|\sin\theta}{\sinh r},
\end{equation}
and, for $1\le p<\infty$ and $z\ne\pm i$,
\begin{equation}\label{eq:Lp}
\omega_z\in L^p(\Lambda^1\HH^3)\quad\Longleftrightarrow\quad
p>2\ \text{ and }\ |\im z|<1-\tfrac2p .
\end{equation}
The same holds for $\nabla\omega_z$, since
$|\nabla\omega_z|\le C(|\varphi_z'(r)|+|\varphi_z(r)|)/\sinh r$ and
$\varphi_z'$ grows at the same exponential rate as $\varphi_z$.
\end{lemma}

\begin{proof}
That $\varphi_z$ solves~\eqref{eq:ode} is a direct verification, which we
also checked symbolically. Evenness and smoothness on $\R$ follow from the
evenness of $\cos(zr)$ and $\coth(r)\sin(zr)$ and the expansion
$\coth r\sin zr=z+\tfrac16(2z-z^3)r^2+O(r^4)$, which together with
$z\cos zr=z-\tfrac12z^3r^2+O(r^4)$ gives the stated Taylor coefficients. The leading coefficient $z(z^2+1)$ vanishes only for
$z\in\{0,\pm i\}$; at $z=0$ the normalized limit $\varphi_0=\lim_{z\to0}
\varphi_z/z$ is not zero, and at $z=\pm i$ one checks $\varphi_{\pm i}\equiv0$.

In geodesic normal coordinates $(x,y,w)$ about the pole, with $r^2=x^2+y^2+w^2$,
one has $r^2\eta=x\,dy-y\,dx$, so $\omega_z=r^{-2}\varphi_z(r)\,(x\,dy-y\,dx)$
with $r^{-2}\varphi_z(r)$ a smooth even function of $r$, hence a smooth
function on $\HH^3$. Thus $\omega_z$ is smooth. The dual field is
$u_z=g^{\phi\phi}\varphi_z\sin^2\theta\,\partial_\phi=\varphi_z\sinh^{-2}r\,
\partial_\phi$, which is divergence-free because its coefficient does not
depend on $\phi$ and $\partial_\phi$ is a Killing field; so $\delta\omega_z=0$.

For a function $\varphi$ of $r$ alone, $\delta(\varphi\eta)=0$ as just
shown, and a direct computation in the warped product
$dr^2+\sinh^2r\,g_{S^2}$ gives
\begin{equation}\label{eq:warped}
\Delta_H(\varphi\eta)=\delta d(\varphi\eta)=
\Bigl(-\varphi''+\frac{2}{\sinh^2r}\varphi\Bigr)\eta ,
\end{equation}
where $2$ is the eigenvalue of $\Delta_{S^2}$ on $\eta$. Together
with~\eqref{eq:ode} this gives~\eqref{eq:eig}, the second form
by~\eqref{eq:LD}. The formula for $|\omega_z|_g$ is immediate from
$g^{\phi\phi}=\sinh^{-2}r\sin^{-2}\theta$.

For~\eqref{eq:Lp}, write $\beta:=\im z$. Since $\coth r-1=2e^{-2r}/(1-e^{-2r})$,
\begin{equation}\label{eq:asym}
\varphi_z(r)=\tfrac12(z+i)e^{izr}+\tfrac12(z-i)e^{-izr}+\rho_z(r),\qquad
|\rho_z(r)|\le3e^{(|\beta|-2)r}\quad(r\ge1).
\end{equation}
With $dV=\sinh^2r\sin\theta\,dr\,d\theta\,d\phi$,
\begin{equation}\label{eq:norm}
\|\omega_z\|_{L^p}^p=c_p\int_0^\infty|\varphi_z(r)|^p\sinh^{2-p}r\,dr,\qquad
c_p:=\int_{S^2}\sin^p\theta\,d\sigma\in(0,\infty).
\end{equation}
Near $r=0$ the integrand is $O(r^{p+2})$. For $z\notin\{0,\pm i\}$ and
$\beta\ne0$, \eqref{eq:asym} gives $|\varphi_z(r)|\ge c\,e^{|\beta|r}$ for
large $r$, with $c>0$, and $|\varphi_z(r)|\le Ce^{|\beta|r}$ for all
$r\ge1$; so the integral converges if and only if $p|\beta|+2-p<0$. For
$\beta=0$, $z\ne0$, the function $z\cos zr-\sin zr$ is a non-trivial
oscillation of fixed amplitude, so the integral converges if and only if
$p>2$. For $z=0$, $\varphi_0(r)\sim-r$ and the same conclusion holds.
The criterion~\eqref{eq:Lp} was also checked numerically for nine choices of
$(p,z)$, including cases close to the threshold $|\im z|=1-2/p$ and the
case $p=2$, by comparing the tail integrals of~\eqref{eq:norm} over
successive windows in $r$; the script is available at~\cite{Zen}.
\end{proof}

\begin{remark}
The same construction works for every spherical harmonic $Y$ of degree
$\ell\ge1$, with $\eta$ replaced by $\star_{S^2}dY$ and the potential
$2/\sinh^2r$ in~\eqref{eq:ode} replaced by $\ell(\ell+1)/\sinh^2r$. For
$\ell=2$, for instance, $e^{izr}(-z^2-3iz\coth r+3\coth^2r-1)$ solves the
equation. The membership criterion~\eqref{eq:Lp} is the same for all
$\ell$, since it depends only on the growth $e^{|\im z|r}$ of the
solutions and on the factor $\sinh^{2-p}r$. The coefficient of the
growing exponential is a nonzero multiple of $\prod_{k=1}^{\ell}(z-ik)$ for
$\im z>0$ and of $\prod_{k=1}^{\ell}(z+ik)$ for $\im z<0$, so it vanishes
only where $|\im z|\ge1$, outside the region $|\im z|<1-2/p$.
\end{remark}

\begin{lemma}\label{lem:eigenvector}
Let $p>2$ and $|\im z|<1-2/p$. Then $u_z\in D(A_p)$, $A_pu_z=(4+z^2)u_z$,
and $S_p(t)u_z=e^{-(4+z^2)t}u_z$ for $t\ge0$.
\end{lemma}

\begin{proof}
By Lemma~\ref{lem:swirl} and Lemma~\ref{lem:P}, $\omega_z\in\Lp$. Let
$v\in C_c^\infty(\Lambda^1)$ and $h(t):=\ip{\omega_z}{e^{-t\Delta_H}v}$.
Since $e^{-t\Delta_H}v\in D(\Delta_H^{(p')})$ and $\omega_z$ is smooth with
$\omega_z,\nabla\omega_z,\Delta_H\omega_z\in L^p$ by Lemma~\ref{lem:swirl},
Lemma~\ref{lem:green} and~\eqref{eq:green} give
$h'(t)=-\ip{\omega_z}{\Delta_H^{(p')}e^{-t\Delta_H}v}=
-\ip{\Delta_H\omega_z}{e^{-t\Delta_H}v}=-z^2h(t)$, so $h(t)=e^{-z^2t}h(0)$.
By the symmetry of the semigroup, $\ip{e^{-t\Delta_H}\omega_z}{v}=
e^{-z^2t}\ip{\omega_z}{v}$ for all $v\in C_c^\infty$, hence
$e^{-t\Delta_H}\omega_z=e^{-z^2t}\omega_z$. The claims follow
from~\eqref{eq:S} and the definition of the generator.
\end{proof}

\subsection{Proof of Theorem~\ref{thm:spec}}

Take $a=1$.

\emph{The inclusion $\sigma(A_p)\subset4+Q_p$.} By
\cite[Theorem~1.2]{CL25}, applied to $\HH^3=\HH^{N+1}$ with $N=2$ and to
$1$-forms, $k=1\le N/2$, the spectrum of $\Delta_H^{(p)}$ on
$L^p(\Lambda^1\HH^3)$ is
\begin{equation}\label{eq:CL}
\Bigl\{\bigl(\tfrac N2-k\bigr)^2+z^2:\ |\im z|\le\bigl|\tfrac Np-\tfrac N2\bigr|\Bigr\}
=\bigl\{z^2:\ |\im z|\le|\tfrac2p-1|\bigr\}=Q_p .
\end{equation}
If $\zeta\notin Q_p$, the resolvent $(\zeta-\Delta_H^{(p)})^{-1}$ is bounded
on $L^p$ and commutes with $\PP$ by Lemma~\ref{lem:P}, so it maps $\Lp$
into itself and inverts $\zeta-(A_p-4)$ there. Hence
$\sigma(A_p-4)\subset Q_p$.

\emph{The inclusion $\sigma(A_p)\supset4+Q_p$ for $p>2$.} By
Lemma~\ref{lem:eigenvector}, every $4+z^2$ with $|\im z|<1-2/p$ is an
eigenvalue of $A_p$. These points form the interior of $4+Q_p$, and the
spectrum is closed.

\emph{The inclusion for $1<p<2$.} Let $|\im z|<1-2/p'=2/p-1$. Then
$u_z\in L^{p'}_\sigma$ is an eigenvector of $A_{p'}$ by
Lemma~\ref{lem:eigenvector}, since $p'>2$. For $v\in D(A_p)$ the symmetry
of the semigroup gives $\ip{S_p(t)v}{u_z}=\ip{v}{S_{p'}(t)u_z}=
e^{-(4+z^2)t}\ip{v}{u_z}$, and differentiating at $t=0$ yields
$\ip{(4+z^2-A_p)v}{u_z}=0$. The functional $v\mapsto\ip{v}{u_z}$ is not
zero on $\Lp$, for if it vanished on $\Lp$ then for every $v\in L^p$ we
would have $\ip{v}{u_z}=\ip{\PP v}{u_z}+\ip{v}{(I-\PP)u_z}=0$, hence
$u_z=0$. So $4+z^2-A_p$ is not surjective, and $4+z^2\in\sigma(A_p)$.
Again these points fill the interior of $4+Q_p$, and the spectrum is
closed.

\emph{The case $p=2$.} Here $Q_2=[0,\infty)$. By Donnelly's theorem
\cite[Theorem~1]{Don81} the $L^2$ spectrum of $\Delta_H$ on coclosed
1-forms of $\HH^3$ is $[0,\infty)$ (and on exact 1-forms it is
$[1,\infty)$), so $\sigma(A_2)=[4,\infty)$.

\emph{Eigenvalues.} For $p>2$ the eigenvalues were exhibited above. This
agrees with \cite[Theorem~1.2]{CL25}, which shows that every interior point
of $Q_p$ is an eigenvalue of $\Delta_H^{(p)}$ on all 1-forms; the forms
$\omega_z$ show that the eigenforms can be taken coclosed. For
$1<p\le2$, \cite[Theorem~1.2]{CL25} states that $\Delta_H^{(p)}$ has no
eigenvalues on $L^p(\Lambda^1\HH^3)$, so neither has its part $A_p-4$.
This is consistent with Lemma~\ref{lem:swirl}, since no $u_z$ lies in
$L^p$ for $p\le2$. \qed

\begin{remark}
The exact part $d(L^p)$ of the Hodge decomposition is also invariant under
$\Delta_H$, since $\Delta_Hd=d\Delta_H$, and the spectrum of $\Delta_H$ on
it is contained in the $L^p$ spectrum $1+Q_p$ of $-\Delta_0$ on functions
\cite{DST88}, which lies inside $Q_p$. So the spectrum of $\Delta_H$ on
all 1-forms of $\HH^3$ is carried by the coclosed part. The boundary
parabola of $Q_p$ is traced by the approximate eigenforms that
\cite[Prop.~6.1]{CL25} build from the coclosed family $y^{\nu}dx_1$,
$\nu=2/p-1+is$, of the half-space model, for which
$\delta(y^\nu dx_1)=0$ and $\Delta_H(y^\nu dx_1)=-\nu^2y^\nu dx_1$
\cite[eqs.~(31), (33), (36)]{CL25}. The eigenfields $u_z$ with purely
imaginary $z=i\beta$, $|\beta|<1-2/p$, have the real eigenvalues
$4-\beta^2$; as $|\beta|$ increases they interpolate between the eigenvalue
$4$ at $\beta=0$ and the leftmost point $\mu_p$ of the spectrum, where they
leave $L^p$.
\end{remark}

\section{Exact decay rates}
\label{sec:rates}

\subsection{Proof of Theorem~\ref{thm:rates}}

Take $a=1$. By Section~\ref{sec:prelim}, $S_p(t)$ is an analytic semigroup
on $\Lp$, so its growth bound $\omega_0:=\lim_{t\to\infty}t^{-1}\log
\|S_p(t)\|$ equals its spectral bound $\sup\{\re\zeta:\zeta\in\sigma(-A_p)\}$
\cite[Cor.~IV.3.12]{EN00}. By Theorem~\ref{thm:spec},
\begin{equation}\label{eq:infre}
\inf\{\re\zeta:\zeta\in\sigma(A_p)\}=4+\inf\{\re z^2:\ |\im z|\le|1-\tfrac2p|\}
=4-\bigl(1-\tfrac2p\bigr)^2=\mu_p ,
\end{equation}
the infimum being attained at $z=\pm i(1-2/p)$. Hence $\omega_0=-\mu_p$,
which gives the upper bound in~\eqref{eq:rates} for every $\epsilon>0$. For
the lower bound, $-\mu_p\in\sigma(-A_p)$, so $e^{-\mu_pt}\in\sigma(S_p(t))$
by the spectral inclusion theorem \cite[Thm.~IV.3.6]{EN00}, and
$\|S_p(t)\|\ge e^{-\mu_pt}$. For $p>2$ the lower bound also follows
directly from Lemma~\ref{lem:eigenvector}, which gives
$\|S_p(t)\|\ge e^{-(4-\beta^2)t}$ for every $\beta<1-2/p$. For $p=2$ the
operator $A_2$ is self-adjoint with spectrum $[4,\infty)$, so
$\|S_2(t)\|=e^{-4t}$. Finally $4-(1-2/p)^2=(p+2)(3p-2)/p^2$, and at $p=3$
this is $35/9$. \qed

\subsection{Comparison with the soft lower bounds}

\begin{proposition}\label{prop:lower}
On $\HH^3$ of curvature $-a^2$, for all $t\ge0$ and $1\le p\le\infty$,
\begin{enumerate}
\item[(i)] $\|e^{t\LD}\|_{L^p\to L^p}\le e^{-(2a^2+4(p-1)a^2/p^2)t}$
(domination);
\item[(ii)] $\|e^{t\LD}\|_{L^p\to L^p}\le e^{-(2a^2+4a^2/\max(p,p'))t}$
(interpolation);
\item[(iii)] $\|e^{t\LD}\|_{L^2\to L^2}=\|S_2(t)\|_{L^2_\sigma\to L^2_\sigma}
=e^{-4a^2t}$.
\end{enumerate}
The bounds (i) and (ii) are inherited by $S_p(t)$ on $\Lp$. At $p=3$ the
rates are $26a^2/9$, $10a^2/3$ and $4a^2$. For every $1<p<\infty$ the exact
rate $\mu_p$ exceeds the rate in (i) by exactly $a^2$, and exceeds the rate
in (ii) by $a^2(p-2)(p+2)/p^2$ for $p\ge2$.
\end{proposition}

\begin{proof}
Take $a=1$. (i) follows from~\eqref{eq:dom} and~\eqref{eq:scalar}. For
(iii), $-\LD=\Delta_H+4\ge4$ on $L^2(\Lambda^1)$ by~\eqref{eq:LD}, and the
value is attained by Theorem~\ref{thm:rates} at $p=2$. For (ii) with
$p\ge2$, interpolate (iii) with the $L^\infty$ contraction $e^{-2t}$ of (i)
by the Riesz-Thorin theorem; with $1/p=(1-\vartheta)/2$ the rate is
$\frac2p\cdot4+(1-\frac2p)\cdot2=2+\frac4p$. For $p<2$ interpolate (iii)
with the $L^1$ contraction $e^{-2t}$ of (i) instead, which gives
$2+4/p'$. The differences are $3+\lambda_p-2-\lambda_p=1$ and
$3+\frac{4(p-1)}{p^2}-2-\frac4p=1-\frac4{p^2}$.
\end{proof}

\begin{remark}[where the domination loses]\label{rem:kato}
The bound (i) sees only the scalar rate $\lambda_p=4(p-1)/p^2$ and the
constant Ricci shift $2a^2$. On $\HH^3$ the rough Laplacian on 1-forms is
$-\Delta_B=\Delta_H+2a^2$ by~\eqref{eq:LD}, and its $L^p$ spectrum is
$a^2(2+Q_p)$ by~\eqref{eq:CL}, with leftmost point
$a^2(2-(1-2/p)^2)=a^2(1+\lambda_p)$. So the $L^p$ spectral bottom of the
rough Laplacian on 1-forms exceeds that of the scalar Laplacian by exactly
$a^2$, for every $p$, and this is what Kato's inequality
$|\nabla|\omega||\le|\nabla\omega|$, on which the domination rests,
discards.
\end{remark}

\begin{remark}[the Hodge choice of viscous operator]\label{rem:hodge}
Theorem~\ref{thm:spec} also decides the behavior of the other operator
that the thin-shell limit can produce~\cite{WB-shell}. If the viscous term
is $-\mu\Delta_Hu$ instead of $\mu\LD u$, the Stokes operator on $\HH^3$
is the part of $\Delta_H$ in $\Lp$, whose spectrum is $a^2Q_p$
by~\eqref{eq:spec}. Its leftmost point is $-a^2(1-2/p)^2<0$ for $p\ne2$,
so the linearized flow $e^{-t\mu\Delta_H}$ has exponential type
$+\mu a^2(1-2/p)^2$ on $\Lp$, and for $p>2$ the swirl fields $u_{i\beta}$,
$0<\beta<1-2/p$, are divergence-free $L^p$ fields which grow like
$e^{\mu a^2\beta^2t}$ under it. In $L^2$ both operators are non-negative,
with spectra $[4a^2,\infty)$ and $[0,\infty)$. So on $\HH^3$ the two
candidates differ, for $p\ne2$, not only in the decay rate but in the sign
of the exponential type of the Stokes semigroup on $\Lp$.
\end{remark}

\begin{remark}[on the constant]\label{rem:eps}
For $p=2$, and for the scalar semigroup for every $p$, the exponential type
is attained with constant $1$ and $\epsilon=0$. We do not know whether
$\|S_p(t)\|_{\Lp\to\Lp}\le Ce^{-\mu_pt}$ holds for $p\ne2$, that is,
whether the decay at the spectral bound is free of polynomial corrections.
The eigenfields give only $\|S_p(t)\|\ge e^{-\mu_pt}$.
\end{remark}

\section{The Navier-Stokes equations}
\label{sec:NS}

Throughout this section the curvature is $-1$, $\mu>0$ is the viscosity,
and $S(t):=S_p(\mu t)$ denotes the Stokes semigroup with viscosity on
whichever $\Lp$ it acts; by consistency this is unambiguous. Constants $C$
depend only on the Lebesgue exponents and on $\lambda$ or $\epsilon$ where
indicated, never on $\mu$, $t$ or the data.

\begin{lemma}[$L^p$-$L^q$ bounds]\label{lem:LpLq}
Let $1<p\le q<\infty$ and $\epsilon>0$. For $u_0\in\Lp$ and $t>0$,
\begin{equation}\label{eq:LpLq}
\|S(t)u_0\|_{L^q}\le C_\epsilon\min(\mu t,1)^{-\frac32(\frac1p-\frac1q)}
e^{-\mu(\mu_p-\epsilon)t}\|u_0\|_{L^p}.
\end{equation}
\end{lemma}

\begin{proof}
For $0<\mu t\le1$, \cite[Theorem~4.1]{Pie17} gives
$\|e^{\tau\LD}\|_{L^p\to L^q}\le C\tau^{-\frac32(\frac1p-\frac1q)}$ for
$0<\tau\le1$; apply it with $\tau=\mu t$. For $\mu t>1$ write
$S(t)=S(1/\mu)S(t-1/\mu)$, bound $S(1/\mu)$ from $L^p$ to $L^q$ by the
short-time bound at $\tau=1$ and $S(t-1/\mu)$ on $\Lp$ by
Theorem~\ref{thm:rates}; the factor $e^{\mu_p-\epsilon}$ so produced is
absorbed in $C_\epsilon$.
\end{proof}

\begin{lemma}[short-time bilinear bound]\label{lem:Pdiv}
Let $1<r<\infty$. For 2-tensor fields $F\in L^r$ and $0<\mu s\le1$,
$\|S(s)\PP\nabla\!\cdot F\|_{L^r}\le C(\mu s)^{-1/2}\|F\|_{L^r}$.
\end{lemma}

\begin{proof}
By Lemma~\ref{lem:P}, $S(s)\PP=\PP e^{s\mu\LD}$ on $L^r$. For
$F\in C_c^\infty$ and $v\in C_c^\infty(\Lambda^1)$, the symmetry of
$e^{s\mu\LD}$ and of $\PP$ gives
$\ip{\PP e^{s\mu\LD}\nabla\!\cdot F}{v}=-\ip{F}{\nabla e^{s\mu\LD}\PP v}$.
By \cite[Theorem~4.2]{Pie17}, $\|\nabla e^{\tau\LD}w\|_{L^{r'}}\le
C\tau^{-1/2}\|w\|_{L^{r'}}$ for $0<\tau\le1$, and $\|\PP v\|_{L^{r'}}\le
C\|v\|_{L^{r'}}$ by Lemma~\ref{lem:P}. Take the supremum over
$\|v\|_{L^{r'}}\le1$ and extend to $F\in L^r$ by density.
\end{proof}

\begin{proposition}[bilinear estimate with the exact rate]\label{prop:bil}
Let $\epsilon>0$ and $q\in\{3,6\}$. For $F\in L^3$ and $\tau>0$,
\begin{equation}\label{eq:bil}
\|S(\tau)\PP\nabla\!\cdot F\|_{L^q}\le C_\epsilon\min(\mu\tau,1)^{-\frac12-\frac32(\frac13-\frac1q)}
e^{-\mu(\mu_3-\epsilon)\tau}\|F\|_{L^3}.
\end{equation}
The exponents are $-1/2$ for $q=3$ and $-3/4$ for $q=6$.
\end{proposition}

\begin{proof}
Let $\tau_0:=\min(\tau,1/\mu)/2$ and factor
$S(\tau)\PP\nabla\!\cdot=S(\tau-\tau_0)\,[S(\tau_0)\PP\nabla\!\cdot\,]$. By
Lemma~\ref{lem:Pdiv} the bracket is bounded on $L^3$ by $C(\mu\tau_0)^{-1/2}$,
and by Lemma~\ref{lem:LpLq} with $p=3$ the outer factor is bounded from
$L^3_\sigma$ to $L^q$ by $C_\epsilon\min(\mu(\tau-\tau_0),1)^{-\frac32(\frac13-\frac1q)}
e^{-\mu(\mu_3-\epsilon)(\tau-\tau_0)}$. Since $\mu\tau_0=\min(\mu\tau,1)/2$,
$\mu(\tau-\tau_0)\ge\min(\mu\tau,1)/2$ and
$e^{\mu(\mu_3-\epsilon)\tau_0}\le e^{\mu_3/2}$, the claim follows.
\end{proof}

\subsection{Proof of Theorem~\ref{thm:NS}}

Fix $0<\lambda<35/9$ and put $\epsilon:=35/9-\lambda>0$, so that
$\mu_3-\epsilon=\lambda$, and $\alpha:=\mu\lambda$.

\emph{Linear term.} By Lemma~\ref{lem:LpLq} with $(p,q)=(3,6)$,
$\|S(\cdot)u_0\|_{X_\lambda}\le C_1\|u_0\|_{L^3}$ with $C_1=C_1(\lambda)$.

\emph{Bilinear term.} Let
$B(u,v)(t):=-\int_0^tS(t-s)\PP\nabla\!\cdot(u\otimes v)(s)\,ds$. Since
$|u\otimes v|=|u||v|$ pointwise,
$\|(u\otimes v)(s)\|_{L^3}\le\min(\mu s,1)^{-1/2}e^{-2\alpha s}
\|u\|_{X_\lambda}\|v\|_{X_\lambda}$. Proposition~\ref{prop:bil} with $q=6$
gives
\begin{equation}\label{eq:Bbound}
\min(\mu t,1)^{1/4}e^{\alpha t}\|B(u,v)(t)\|_{L^6}\le\frac{C_\epsilon}{\mu}\,
I(\mu t)\,\|u\|_{X_\lambda}\|v\|_{X_\lambda},
\end{equation}
where, after the substitution $\sigma=\mu s$ and using
$e^{\alpha t}e^{-\alpha(t-s)}e^{-2\alpha s}=e^{-\alpha s}$,
\begin{equation}\label{eq:I}
I(\tau):=\min(\tau,1)^{1/4}\int_0^\tau\min(\tau-\sigma,1)^{-3/4}
\min(\sigma,1)^{-1/2}e^{-\lambda\sigma}\,d\sigma .
\end{equation}
For $\tau\le1$, $I(\tau)\le\int_0^1(1-x)^{-3/4}x^{-1/2}dx=B(\tfrac14,\tfrac12)$,
the Beta function. For $\tau>1$ split the integral at $\tau/2$. On
$[0,\tau/2]$ we have $\min(\tau-\sigma,1)^{-3/4}\le2^{3/4}$ and
$\int_0^\infty\min(\sigma,1)^{-1/2}e^{-\lambda\sigma}d\sigma\le2+\lambda^{-1}$.
On $[\tau/2,\tau]$ we have $\min(\sigma,1)^{-1/2}\le\sqrt2$,
$e^{-\lambda\sigma}\le e^{-\lambda\tau/2}$ and
$\int_{\tau/2}^\tau\min(\tau-\sigma,1)^{-3/4}d\sigma\le4+\tau/2$, while
$(4+\tau/2)e^{-\lambda\tau/2}\le4+(e\lambda)^{-1}$. Hence
\begin{equation}\label{eq:supI}
\sup_{\tau>0}I(\tau)\le\max\Bigl\{B(\tfrac14,\tfrac12),\
2^{3/4}(2+\lambda^{-1})+\sqrt2\,\bigl(4+(e\lambda)^{-1}\bigr)\Bigr\}=:J_\lambda,
\end{equation}
and $\|B(u,v)\|_{X_\lambda}\le C_2\mu^{-1}\|u\|_{X_\lambda}\|v\|_{X_\lambda}$
with $C_2:=C_\epsilon J_\lambda$.

\emph{Contraction.} The map $u\mapsto S(\cdot)u_0+B(u,u)$ is a contraction
of the ball $\{\|u\|_{X_\lambda}\le2C_1\|u_0\|_{L^3}\}$ into itself as soon
as $4C_1C_2\mu^{-1}\|u_0\|_{L^3}<1$, that is for
$\|u_0\|_{L^3}<\epsilon_0:=\mu/(4C_1C_2)$ \cite[Lemma~5.3]{Pie17}. Put
$C_\lambda:=2C_1$ and $c_\lambda:=1/(4C_1C_2)$. The fixed point is the
unique solution of~\eqref{eq:mild} in the ball of Theorem~\ref{thm:NS},
and the $L^6$ bound in~\eqref{eq:NSdecay} is
$\|u\|_{X_\lambda}\le C_\lambda\|u_0\|_{L^3}$.

\emph{Decay of the $L^3$ norm.} By Lemma~\ref{lem:LpLq} with $p=q=3$,
$\|S(t)u_0\|_{L^3}\le C_\epsilon e^{-\alpha t}\|u_0\|_{L^3}$. By
Proposition~\ref{prop:bil} with $q=3$,
\begin{equation}
\begin{aligned}
\|B(u,u)(t)\|_{L^3}&\le C_\epsilon e^{-\alpha t}\|u\|_{X_\lambda}^2
\int_0^t\min(\mu(t-s),1)^{-1/2}\min(\mu s,1)^{-1/2}e^{-\alpha s}\,ds\\
&\le\frac{C_\epsilon J'_\lambda}{\mu}\,e^{-\alpha t}\|u\|_{X_\lambda}^2,
\end{aligned}
\end{equation}
where $J'_\lambda:=\max\{\pi,\ \sqrt2(2+\lambda^{-1})+\sqrt2(2+(e\lambda)^{-1})\}$
bounds the integral after the substitution $\sigma=\mu s$ by the same
splitting as above, the value $\pi=B(\frac12,\frac12)$ covering $\mu t\le1$.
Since $\|u\|_{X_\lambda}\le C_\lambda\|u_0\|_{L^3}$ and
$\|u_0\|_{L^3}<c_\lambda\mu$, the $L^3$ bound in~\eqref{eq:NSdecay} follows
with $C'_\lambda:=C_\epsilon(1+c_\lambda C_\lambda^2J'_\lambda)$, which
depends only on $\lambda$.

\emph{Sharpness.} If $\|S(t)u_0\|_{L^3}\le Ce^{-\mu\lambda t}\|u_0\|_{L^3}$
held for all $u_0\in L^3_\sigma$ with some $\lambda>35/9$, the exponential
type of $S_3$ would be at most $-\lambda<-\mu_3$, contradicting
Theorem~\ref{thm:rates}. \qed

\begin{remark}
The constants $C_1,C_2$ and hence $\epsilon_0/\mu$ depend on $\lambda$
through $C_\epsilon$ in Lemma~\ref{lem:LpLq} and through $J_\lambda$. As
$\lambda\uparrow35/9$ the admissible data shrink, unless the question of
Remark~\ref{rem:eps} has a positive answer, in which case $\lambda=35/9$
itself is admissible. For the explicit rates $\lambda=26/9$ and
$\lambda=10/3$ of Proposition~\ref{prop:lower} the semigroup bounds hold
with constant $1$, so that $C_\epsilon$ in Lemma~\ref{lem:LpLq} is then
the short-time constant alone.
\end{remark}

\section{Scaling, the critical exponent, and remarks}
\label{sec:scaling}

\paragraph{Scaling.}
Let $u$ solve~\eqref{eq:NS} on $(\HH^3,g)$ with viscosity $\mu$ and
pressure $p$. For $\kappa>0$ put $\tilde g:=\kappa^2g$,
$\tilde u(\tilde t):=\kappa^{-2}u(\kappa^{-2}\tilde t)$ as a vector field,
and $\tilde p:=\kappa^{-2}p$. The Levi-Civita connection of $\tilde g$ is
that of $g$, the operator $\LD$ of $\tilde g$ is $\kappa^{-2}$ times that of
$g$, and so is the gradient. Hence $\tilde u$ solves~\eqref{eq:NS} on
$(\HH^3,\tilde g)$, which has curvature $-a^2/\kappa^2$, with the same
viscosity, and
\begin{equation}\label{eq:scalingnorm}
\|\tilde u(\tilde t)\|_{L^q(\tilde g)}=\kappa^{3/q-1}\|u(\kappa^{-2}\tilde t)\|_{L^q(g)} .
\end{equation}
The $L^3$ norm is invariant. Consequently Theorem~\ref{thm:NS} holds for
curvature $-a^2$ in the form
\begin{equation}\label{eq:NSa}
\|u(t)\|_{L^6}\le C_\lambda a^{1/2}\|u_0\|_{L^3}\min(a^2\mu t,1)^{-1/4}
e^{-\mu\lambda a^2t},\qquad
\|u(t)\|_{L^3}\le C'_\lambda\|u_0\|_{L^3}\,e^{-\mu\lambda a^2t},
\end{equation}
for $0<\lambda<35/9$ and $\|u_0\|_{L^3}<c_\lambda\mu$, with the same
constants as for $a=1$. Curvature sets the time scale $1/(\mu a^2)$ of the
decay, not the size of the admissible $L^3$ data. For data in $L^q$ with
$q>3$ the dimensionless size of the datum is $\|u_0\|_{L^q}/(\mu a^{1-3/q})$,
which decreases as $a$ grows, so larger curvature admits larger data; this
is the mechanism behind \cite[Theorem~2.1]{HH24arch}.

\paragraph{The critical exponent.}
The exponent $3$ in Theorem~\ref{thm:NS} is the usual scaling threshold
\cite{Wei80,Kat84,Gig86}. On $\R^3$ the equations are invariant under
$u\mapsto\kappa u(\kappa x,\kappa^2t)$, which preserves $\|u_0\|_{L^3}$. On
$\HH^3$ the scaling above has the same effect on norms and maps solutions
to solutions on the rescaled space. Below the length scale $1/a$ the
manifold is nearly Euclidean, and under bounded geometry the short-time
singularities of $S(t)$ and of $S(t)\PP\nabla\cdot$ are the Euclidean
ones (Lemma~\ref{lem:LpLq} and Proposition~\ref{prop:bil}). If one repeats
the contraction argument for $L^p$ data, $3/2<p<3$, in the norm
$\sup_t\min(\mu t,1)^{\gamma}e^{\alpha t}\|u(t)\|_{L^q}$ with
$\gamma=\frac32(\frac1p-\frac1q)<\frac12$ and $q>3$, then for $\mu t\le1$
the bilinear term is bounded, relative to the linear term, by a multiple of
\begin{equation}\label{eq:count}
(\mu t)^{\frac12-\frac3{2p}}\,B\bigl(\tfrac12-\tfrac3{2q},\,1-2\gamma\bigr),
\end{equation}
because $\gamma+\frac12-\frac3{2q}-2\gamma=\frac12-\frac3{2p}$. For $p<3$
this diverges as $t\to0$, and Kato's scheme does not close, exactly as on
$\R^3$; for $p\le3/2$ no choice of $q>3$ even makes the Beta function
finite. The spectral gap enters only through $e^{\alpha t}$, which is
bounded as $t\to0$, so curvature does not move the threshold. This concerns
the method and does not show that small $L^p$ data with $p<3$ fail to
produce global solutions.

\paragraph{Beyond hyperbolic space.}
On a complete simply connected 3-manifold with pinched negative curvature
$-b^2\le K\le-a^2<0$ that is not a space form, the Leray projector no
longer commutes with $\LD$, the Stokes operator is no longer $\Delta_H+4a^2$
on coclosed forms, and no exact spectral information is available. The
domination~\eqref{eq:dom} and the interpolation of
Proposition~\ref{prop:lower}(ii) survive, and give the rates $26a^2/9$ and
$10a^2/3$ for $e^{t\LD}$ on $L^3$; the Stokes semigroup then decays at a
rate reduced by the $L^3$ norm of the commutator, which is proportional to
$b^2-a^2$ \cite{WB-pinched}. Theorem~\ref{thm:rates} shows how much room
there is between these soft bounds and the truth already in the model case.
The $L^2$ gap of the deformation Laplacian plays the same role for the
stochastic equations with thermal noise. On compact manifolds with
$\Ric\le-\kappa^2g$ the Gibbs equilibrium is approached at a rate at least
$2\mu\lambda_\Def$ with $\lambda_\Def\ge\kappa^2$~\cite{BW-therm}, and
Proposition~\ref{prop:lower}(iii) gives the exact value $4a^2$ of that gap
on $\HH^3$.

\paragraph{Open questions.}
Whether the exact rate is attained with a constant, Remark~\ref{rem:eps},
is open for $p\ne2$. Whether the small-data solutions of
Theorem~\ref{thm:NS} decay at exactly the linear rate, rather than merely
at every slower rate, would follow from an asymptotic expansion of
$u(t)-S(t)u_0$ and is not addressed here.

\section*{Declarations}

\textbf{Conflict of interest.}
The authors declare that they have no conflict of interest.\\
\textbf{Data and code availability.}
The data and code are openly available on Zenodo
at \href{https://doi.org/10.5281/zenodo.24025899}{doi:10.5281/zenodo.24025899}~\cite{Zen}.


\begin{thebibliography}{99}\raggedright

\bibitem{Bal20}
B.~Balentine, Well-posedness and global in time behavior for $L^p$-mild
solutions to the Navier-Stokes equation on the hyperbolic space,
arXiv:2008.01850 (2020).

\bibitem{BW-therm}
S.~L.~Braunstein, Z.-W.~Wang, Exponential thermalisation of viscous fluids
on negatively curved manifolds, arXiv:2606.02286 (2026).

\bibitem{CCD17}
C.~H.~Chan, M.~Czubak, M.~M.~Disconzi, The formulation of the Navier-Stokes
equations on Riemannian manifolds, J.\ Geom.\ Phys.\ \textbf{121} (2017)
335--346, doi:10.1016/j.geomphys.2017.07.015.

\bibitem{CL25}
N.~Charalambous, Z.~Lu, $L^p$-spectral theory for the Laplacian on forms,
J.\ Funct.\ Anal.\ \textbf{289} (2025), no.~6, Paper No.~110976,
doi:10.1016/j.jfa.2025.110976; arXiv:2401.02136.

\bibitem{Cow83}
M.~Cowling, Harmonic analysis on semigroups, Ann.\ of Math.\ (2)
\textbf{117} (1983), no.~2, 267--283, doi:10.2307/2007077.

\bibitem{Czu24}
M.~Czubak, In search of the viscosity operator on Riemannian manifolds,
Notices Amer.\ Math.\ Soc.\ \textbf{71} (2024), no.~1, 8--16,
doi:10.1090/noti2840.

\bibitem{DST88}
E.~B.~Davies, B.~Simon, M.~Taylor, $L^p$ spectral theory of Kleinian
groups, J.\ Funct.\ Anal.\ \textbf{78} (1988), no.~1, 116--136,
doi:10.1016/0022-1236(88)90135-8.

\bibitem{Don81}
H.~Donnelly, The differential form spectrum of hyperbolic space,
Manuscripta Math.\ \textbf{33} (1981) 365--385, doi:10.1007/BF01798234.

\bibitem{EM70}
D.~G.~Ebin, J.~Marsden, Groups of diffeomorphisms and the motion of an
incompressible fluid, Ann.\ of Math.\ (2) \textbf{92} (1970), no.~1,
102--163, doi:10.2307/1970699.
\bibitem{EN00}
K.-J.~Engel, R.~Nagel, One-Parameter Semigroups for Linear Evolution
Equations, Grad.\ Texts in Math.\ \textbf{194}, Springer, New York, 2000,
doi:10.1007/b97696.
\bibitem{FK64}
H.~Fujita, T.~Kato, On the Navier-Stokes initial value problem.~I, Arch.\
Rational Mech.\ Anal.\ \textbf{16} (1964), no.~4, 269--315,
doi:10.1007/BF00276188.
\bibitem{Gig86}
Y.~Giga, Solutions for semilinear parabolic equations in $L^p$ and
regularity of weak solutions of the Navier-Stokes system, J.\ Differential
Equations \textbf{62} (1986), no.~2, 186--212,
doi:10.1016/0022-0396(86)90096-3.

\bibitem{HH24arch}
T.~H.~Nguyen, T.~N.~H.~Vu, Ricci curvature and the size of initial data for
the Navier-Stokes equations on Einstein manifolds, Arch.\ Math.\ (Basel)
\textbf{122} (2024), no.~1, 83--93, doi:10.1007/s00013-023-01923-5.

\bibitem{Heb99}
E.~Hebey, Nonlinear Analysis on Manifolds: Sobolev Spaces and Inequalities,
Courant Lect.\ Notes Math.\ \textbf{5}, Amer.\ Math.\ Soc., Providence, 1999,
doi:10.1090/cln/005.

\bibitem{HXHM21}
T.~H.~Nguyen, T.~X.~Pham, T.~N.~H.~Vu, T.~M.~Vu, Periodic solutions to
Navier-Stokes equations on non-compact Einstein manifolds with negative
curvature, Anal.\ Math.\ Phys.\ \textbf{11} (2021), no.~2, Paper No.~60,
doi:10.1007/s13324-021-00488-2.

\bibitem{Kat84}
T.~Kato, Strong $L^p$-solutions of the Navier-Stokes equation in $\R^m$,
with applications to weak solutions, Math.\ Z.\ \textbf{187} (1984), no.~4,
471--480, doi:10.1007/BF01174182.
\bibitem{Loh85}
N.~Lohou\'e, Comparaison des champs de vecteurs et des puissances du
laplacien sur une vari\'et\'e riemannienne \`a courbure non positive, J.\
Funct.\ Anal.\ \textbf{61} (1985), no.~2, 164--201,
doi:10.1016/0022-1236(85)90033-3.
\bibitem{Pie17}
V.~Pierfelice, The incompressible Navier-Stokes equations on non-compact
manifolds, J.\ Geom.\ Anal.\ \textbf{27} (2017), no.~1, 577--617,
doi:10.1007/s12220-016-9691-1.

\bibitem{Ste70}
E.~M.~Stein, Topics in Harmonic Analysis Related to the Littlewood-Paley
Theory, Ann.\ of Math.\ Stud.\ \textbf{63}, Princeton Univ.\ Press,
Princeton, 1970; electronic reissue doi:10.1515/9781400881871.
\bibitem{Str83}
R.~S.~Strichartz, Analysis of the Laplacian on the complete Riemannian
manifold, J.\ Funct.\ Anal.\ \textbf{52} (1983), no.~1, 48--79,
doi:10.1016/0022-1236(83)90090-3.
\bibitem{WB-sel}
Z.-W.~Wang, S.~L.~Braunstein, Resolving the viscosity operator ambiguity
on Riemannian manifolds via a kinematic selection principle,
arXiv:2605.17502 (2026).

\bibitem{WB-shell}
Z.-W.~Wang, S.~L.~Braunstein, Boundary conditions select the viscous
operator on Riemannian hypersurfaces: formal analysis and rigorous
thin-shell limits, arXiv:2605.20589 (2026).

\bibitem{WB-pinched}
Z.-W.~Wang, S.~L.~Braunstein, Exponential stability for the
three-dimensional Navier-Stokes equations on negatively curved manifolds,
arXiv:2606.04407 (2026).

\bibitem{Zen}
Z.-W.~Wang, S.~L.~Braunstein, Verification script \texttt{h3\_swirl\_fixed.py}
for the $L^p$ membership of swirl eigenfields on $\HH^3$, Zenodo (2026),
doi:10.5281/zenodo.24025899.

\bibitem{Wei80}
F.~B.~Weissler, The Navier-Stokes initial value problem in $L^p$, Arch.\
Rational Mech.\ Anal.\ \textbf{74} (1980), no.~3, 219--230,
doi:10.1007/BF00280539.

\end{thebibliography}
\end{document}